\documentclass[11pt, a4paper]{article}
\usepackage{amsmath, amssymb, amsfonts, amsthm}
\usepackage[pdftex]{graphicx, color}

\newcommand{\E}{{\bf{E}}}
\newcommand{\PP}{{\bf{P}}}

\newtheorem{tm}{Theorem}
\newtheorem{lem}{Lemma}

\begin{document}

\parindent=0pt

\smallskip
\centerline{\LARGE \bfseries Correlation between clustering and degree}
\smallskip
\centerline{\LARGE \bfseries in affiliation networks}

\par\vskip 3.5em
\centerline{Mindaugas Bloznelis and Justinas Petuchovas}
\vskip 0.5truecm

\centerline{Faculty of Mathematics and Informatics, Vilnius University,}
\centerline{Naugarduko 24, Vilnius 03225, Lithuania}
\vskip 1truecm
\centerline{E-mail: {\it mindaugas.bloznelis@mif.vu.lt}}
\vskip 1truecm

\begin{abstract}We are interested in the probability
that two randomly 
selected neighbors
of a random vertex of degree (at least) $k$ are  adjacent. 
We 
evaluate this 
probability for a power law random intersection graph, where  each vertex is 
prescribed a collection of attributes and two vertices are 
adjacent whenever they
share a common attribute. 
We show that the probability obeys the scaling
$k^{-\delta}$ as $k\to+\infty$. Our results are mathematically rigorous.
The parameter $0\le \delta\le 1$ 
is determined by the tail indices of power law random weights 
defining the links between vertices and 
attributes.
\end{abstract}

\bigskip
{\bf Keywords}:   clustering coefficient, degree distribution, random intersection graph, affiliation network, complex network.

\bigskip

\section{Introduction and Results}

It looks plausible, that in a social network the chances of two neighbors of 
a given actor to be adjacent 
is a decreasing function of actor's degree 
(the total number of its neighbors).
Empirical evidence  of this phenomenon has been reported in a number of papers, see, e.g.,
\cite{Eckmann2002},
\cite{Vazquez2002}, \cite{RavaszB2003}, \cite{Foudalis2011}. 
Theoretical explanations 
 have been derived 
in  \cite{DorogovtsevGM2002} and \cite{RavaszB2003} with the aid of 
 a hierarchical deterministic network model,  and in \cite{Bloznelis2013} with the aid
of  
a random intersection graph model of an affiliation network.
We note that theoretical results \cite{DorogovtsevGM2002}, \cite{RavaszB2003} and
 \cite{Bloznelis2013} only address the scaling $k^{-1}$, i.e., $\delta=1$. In particular, they do not explain empirically observed   scaling $k^{-\delta}$ with $\delta\approx 0.75$ reported in  \cite{Vazquez2002}, see also 
 \cite{Foudalis2011}.
  In the 
present paper we develop further the approach of  
 \cite{Bloznelis2013} and address  the range
$0\le \delta<1$. The development resorts to a more realistic fitness model of an 
affiliation network that accounts for variable activities of actors and 
attractiveness of
attributes described below.

An affiliation network defines adjacency relations between actors by using
an auxiliary set of attributes. Let $V=\{v_1,\dots, v_n\}$ denote the set 
of actors (vertices)
and $W=\{w_1,\dots, w_m\}$ denote the set of attributes. 
Every actor $v_i$ is 
prescribed a collection  of attributes and two actors $v_i$ and $v_j$
are declared adjacent in the network if they share a common attribute.
For example, in the 
film actor network two actors are adjacent if they 
have played in the same movie, in 
the collaboration network two scientists are adjacent if 
they have coauthored a publication, in the consumer copurchase network two 
consumers are adjacent if they have purchased similar products.

A convenient model of a large affiliation network is obtained by 
linking (prescribing) attributes to actors at random \cite{godehardt2001},
\cite{karonski1999},
\cite{NewmanSW2001}).  
Furthermore, in order to model the  heterogeneity of human activity, 
 we assign every actor $v_j$ a random 
weight $Y_j$ reflecting its activity. Similarly, 
a random  weight $X_i$ is assigned 
to an attribute $w_i$ to model its attractiveness. 
Now $w_i$ is linked to 
 $v_j$ at random 
and with probability proportional to the attractiveness $X_i$ and activity $Y_j$. 
The random affiliation network obtained in this way is called
a  random intersection graph, see  
\cite{BloznelisGJKR2015}. 

We assume in what follows that $X_0, X_1, \dots, X_m, Y_0,Y_1, \dots, Y_n$ are
 independent non-negative random 
variables. Furthermore, each $X_i$ (respectively $Y_j$) has the  same 
probability distribution denoted 
$P_X$ (respectively $P_Y$). Given realized values 
$X = \{X_i\}_{i=1}^m$ and $Y = \{Y_j\}_{j=1}^n$ 
we define the random bipartite
graph $H_{X,Y}$ with the bipartition  $W\cup V$, 
where links $\{w_i, v_j\}$ are inserted with
probabilities $p_{ij} = \min\{1, X_i Y_j / \sqrt{nm}\}$ independently for each
$(i,j) \in [m] \times [n]$. The random intersection graph
${\cal G}=G(P_X, P_Y, n, m)$ defines the adjacency relation on the vertex set $V$:
vertices $v', v'' \in V$ are declared adjacent (denoted $v' \sim v''$) whenever
$v'$ and $v''$ have a common neighbor in $H_{X,Y}$. 
Such a neighbor belongs to the set $W$ and it is called a witness 
of the edge $v'\sim v''$. We note that for $n,m\to+\infty$ 
satisfying $m/n\to\beta$ 
for some $\beta>0$, 
the random intersection graph ${\cal G}$ admits a tunable global 
clustering coefficient and power 
law degree distribution 
  \cite{BloznelisDamarackas2013},
\cite{BloznelisKurauskas2017}.

Next we introduce
 network characteristics studied in this paper.
Given a finite graph $G$ and integer $k=2,3,\dots$,
%
%
%
define the clustering
coefficients
\begin{eqnarray}\label{=k}
 c_G(k)
&
=
& 
\PP
\bigl(
v_2^*\sim v_3^*\bigr| v_2^*\sim v_1^*, v_3^*\sim v_1^*, d(v_1^*)=k
\bigr),
\\
\label{>=k}
C_G(k)
&
=
& 
\PP
\bigl(
v_2^*\sim v_3^*\bigr| v_2^*\sim v_1^*, v_3^*\sim v_1^*, d(v_1^*)\ge k
\bigr).
\end{eqnarray}
Here $(v_1^*, v_2^*, v_3^*)$ is an ordered triple 
of vertices of $G$ drawn uniformly at random,
  $d(v)$ denotes 
the degree of a vertex $v$. Note that for a 
{\it deterministic} graph $G$,
coefficients (\ref{=k}) and (\ref{>=k})
are the respective ratios of  subgraph counts
\begin{equation}\label{2017-06-04}
\frac
{
\sum_{v:\, d(v)=k} N_{\Delta}(v)
} 
{
\sum_{v:\, d(v)=k}{{d(v)}\choose{2}}
}
\qquad
\
{\rm{and}}
\qquad
\
\frac
{
\sum_{v:\, d(v)\ge k} N_{\Delta}(v)
} 
{
\sum_{v:\, d(v)\ge k}{{d(v)}\choose{2}}.
}
\end{equation}
 Here $N_{\Delta}(v)$ and ${{d(v)}\choose {2}}$
are the numbers of triangles and cherries incident to $v$. Differently, for the {\it random} graph 
${\cal G}$ the conditional probabilities (\ref{=k}) and (\ref{>=k}) refer to the 
two 
sources of randomness: the random sampling of vertices  $(v_1^*, v_2^*, v_3^*)$ 
and the randomly graph generation mechanism. From the fact that the probability 
distribution  of ${\cal G}$ is invariant under permutation of its 
vertices we obtain 
that 
\begin{eqnarray}\label{2017-06-04+1}
c_{\cal G}(k)
&=&
\PP
\bigl(
v_2\sim v_3\bigr| v_2\sim v_1, v_3\sim v_1, d(v_1)=k
\bigr),
\\
\label{2017-06-04+2}
C_{\cal G}(k)
&=&
\PP
\bigl(
v_2\sim v_3\bigr| v_2\sim v_1, v_3\sim v_1, d(v_1)\ge k
\bigr).
\end{eqnarray}
An argument bearing on the
law of large numbers suggests that for large $n,m$  the ratios (\ref{2017-06-04}) can be approximated by respective  probabilities 
(\ref{2017-06-04+1}) and (\ref{2017-06-04+2}).

Our Theorem 2 below establishes a first order asymptotics as 
$n,m\to+\infty$ of the probabilities (\ref{2017-06-04+1}) and (\ref{2017-06-04+2})
\begin{eqnarray}\label{T1-21-1}
c_{\cal G}(k)
&=&
\left(1+\beta^{1/2} b(k)a^{-1}(k)\right)^{-1}+o(1),
\\
\label{T1-21-2}
C_{\cal G}(k)
 &=& 
\left(1+\beta^{1/2} B(k)A^{-1}(k)\right)^{-1}+o(1).
\end{eqnarray}
Here $a(k), b(k)$ and  $A(k)$, $B(k)$ are defined in Theorem 2 below.
Our Theorem 1 describes the dependence on $k$ of the leading term of (\ref{T1-21-2}). Namely, for a power law distributions $P_X$ and $P_Y$ 
the leading term of (\ref{T1-21-2}) obeys the scaling
 $k^{-\delta}$.

\begin{tm}\label{T0} Let $\alpha,\gamma>5$ and $\beta, c_X, c_Y>0$. Let $m,n \to \infty$. 
Assume that $m/n\to \beta$. 
 Suppose  that as $t\to+\infty$
\begin{equation}
 \PP(X>t)=(c_X+o(1))t^{-\alpha},
\qquad\qquad
\PP(Y>t)=(c_Y+o(1))t^{-\gamma}.
\end{equation}
Then for $\delta=((\alpha-\gamma-1)\wedge 1)\vee (-1)$  we have as $k\to+\infty$
\begin{equation}\label{T0+}
 \frac{B(k)}{A(k)}=(c+o(1))k^{\delta}.
\end{equation}
The constant $c=c(\alpha,\gamma,\beta,c_X,c_Y)>0$ admits an explicit expression in terms of 
$\alpha,\gamma,\beta, 
c_X, c_Y$.
\end{tm}
It follows from (\ref{T0+}) that for large $n$ and $m$ the clustering coefficient
 $C_{\cal G}(k)$ obeys the scaling $k^{-\delta}$, where $0\le \delta\le 1$. 
A related result establishing $k^{-1}$ scaling for $c_{\cal G}(k)$ has been shown
in \cite{Bloznelis2013} in the case where $P_Y$ is heavy tailed 
and $P_X$ is degenerate ($P(X_i=x)=1$ for some $x>0$).

We note the ``phase transition'' in the scaling $k^{-\delta}$ at 
$\alpha=\gamma+2$: for 
$\alpha\ge \gamma+2$ we have $\delta=1$ and for $\alpha<\gamma+2$ we have
$\delta<1$. 
Our explanation of this phenomenon is as follows.
Every attribute $w_i$ forms a clique in ${\cal G}$ induced by vertices 
linked to $w_i$. Given the weight $X_i$ (of $w_i$), the expected size of the clique is proportional to $X_i$.
Now, 
for relatively small $\alpha$ (namely, $\alpha<\gamma +2$) the sequence
$X_1,X_2,\dots, X_m$ contains sufficiently many large weights so that the 
corresponding
large cliques (formed by attributes)
 have a tangible effect on the probability (\ref{>=k}). Indeed,
large cliques may increase the value of (\ref{>=k}) considerably.

The proof of Theorem \ref{T0}
uses known
results about the
tail asymptotics of randomly stopped sums of 
heavy tailed independent
random variables in the case
 where the random number of summands is heavy tailed
\cite{AleskevicieneLeipusSiaulys}. 
Similar results are likely to be true also for 
the local probabilities of randomly stopped sums (work in progress)\footnote{Update: Technical report {\it "Local probabilities of randomly stopped sums
of power law lattice random variables"} available at 
http://arxiv.org/abs/1801.01035}.  They would extend Theorem 1  to 
 $c_{\cal G}(k)$ as well.

Before formulating Theorem 2 we introduce some more notation.
We denote $a_r = \E X_0^r, \, b_r = \E Y_0^r$. 
Let $\beta\in (0,+\infty)$.
Let $\Lambda_k$, $k=0,1,2$ be mixed Poisson random variables
with the distributions 
\begin{displaymath}
\PP(\Lambda_k=s)=\E e^{-\lambda_k}\lambda_k^s/s!,
\qquad
 s=0,1,\dots.
 \end{displaymath}
Here $\lambda_0=Y_1\beta^{1/2}a_1$ and 
$\lambda_k=X_k\beta^{-1/2}b_1$ for $k=1,2$. 
Furthermore, for $r=0,1,2,\dots$ and $k=0,1,2$, 
let $\Lambda_k^{(r)}$  be a non-negative integer valued random variable with
the distribution
\begin{displaymath}
\PP(\Lambda_k^{(r)}=s)
=
\bigl(
\E \lambda^r_k
\bigr)^{-1}
\E\Bigl( e^{-\lambda_k}\lambda_k^{s+r}/s!\Bigr),
\qquad
s=0,1,2,\dots.
\end{displaymath}
Note that $\Lambda_k^{(0)}$ have the same probability distribution as $\Lambda_k$.
Let $\tau_i$, $i\ge 1$ 
be random  variables with the  probability distribution 
\begin{displaymath}
\PP(\tau_i=s)=\frac{s+1}{\E\Lambda_1}\PP(\Lambda_1=s+1),
\qquad
s=0,1,2\dots.
\end{displaymath}
Assuming that random variables 
$\{\tau_i$, $i\ge 1\}$ are independent of 
$\Lambda_0^{(r)}$ we introduce the random variables
\begin{equation}\label{d-star}
d_*^{(r)}=\sum_{j=1}^{\Lambda_0^{(r)}}\tau_j,
\qquad
r=0,1,2.
\end{equation}
We denote for short $d_*=d_*^{(0)}=\sum_{j=1}^{\Lambda_0}\tau_j$.


\begin{tm}\label{T1} Let $m,n \to \infty$. 
Assume that $m/n\to \beta$ for some $\beta\in (0,+\infty)$.
 Suppose  that $\E X_1^4<\infty$ and $\E Y_1^4<\infty$. 
Then for each integer $k\ge 2$ relations 
(\ref{T1-21-1}) and (\ref{T1-21-2}) hold with
\begin{eqnarray}
\nonumber
&&
a(k)
=
a_3b_1^3
\PP\bigl(d_*^{(1)}+\Lambda_1^{(3)}=k-2\bigr),
\ 
\
b(k)
=
a_2^2b_1^2b_2
\PP\bigl(d_*^{(2)}+\Lambda_1^{(2)}+\Lambda_2^{(2)}=k-2\bigr),
\\
\nonumber
&&
A(k)
=
a_3b_1^3
\PP\bigl(d_*^{(1)}+\Lambda_1^{(3)}\ge k-2\bigr),
\ 
\,
B(k)
=
a_2^2b_1^2b_2
\PP\bigl(d_*^{(2)}+\Lambda_1^{(2)}+\Lambda_2^{(2)}\ge k-2\bigr).
\qquad
\end{eqnarray}
Here we assume that random variables $d_*^{(1)}$ and $\Lambda_1^{(3)}$ 
are independent. Furthermore, we assume that random variables $d_*^{(2)}$,
 $\Lambda_1^{(2)}$ and $\Lambda_2^{(2)}$ are independent and 
 $\Lambda_2^{(2)}$ has the same distribution as $\Lambda_1^{(2)}$.
\end{tm}


\section{Proof}
We first prove Theorem \ref{T1} and then Theorem \ref{T0}.
Before the proof  we introduce some notation. We denote 
$\{1,2,\dots, r\}=[r]$ and  $(x)_k=x(x-1)\cdots(x-k+1)$.
We denote by $\{w_i \to v_j\}$  the event that  $w_i$ and $v_j$ 
are neighbors in the bipartite graph  $H = H_{X,Y}$. We denote
\begin{displaymath}
  {\mathbb I}_{ij} = {\mathbb I}_{\{w_i \to v_j\}}, \qquad
  \lambda_{ij}=\frac{X_iY_j}{\sqrt{mn}}.
 \end{displaymath}
Let $\PP^*=\PP_{X_1,Y_1}$ and $\PP^{**}=\PP_{X_1,X_2,Y_1}$ 
denote the 
conditional probabilities given $X_1,Y_1$ and $X_1,X_2,Y_1$ 
respectively. Furthermore, for $i=1,2$, we denote by $\PP_{X_i}$ 
and $\PP_{Y_i}$ the conditional probabilities given $X_i$ and 
$Y_i$ respectively.

{\it Proof of Theorem \ref{T1}.} We only prove 
(\ref{T1-21-1}). The proof of (\ref{T1-21-2}) is much the same.
Introduce events
\begin{displaymath} 
A = \{v_1 \sim v_2, \, v_1 \sim v_3, \, v_2 \sim v_3\}, 
\quad
B = \{v_1 \sim v_2, \, v_1 \sim v_3\}, 
\quad
K = \{d(v_1) = k\}.
\end{displaymath}
We derive (\ref{T1-21-1}) from the identity
\begin{equation}
  \PP(v_2 \sim v_3 \mid v_1 \sim v_2, \, v_1 \sim v_3, \, d(v_1) = k) 
  = \frac{\PP(A \cap K)}{\PP(B \cap K)}
\end{equation}
combined with the relations shown below
\begin{eqnarray}\label{A1+++}
&&
\PP(A \cap K)
=
n^{-2}\beta^{-1/2} a(k)+o(n^{-2}),
\quad
\qquad
\qquad
\\
\label{B1+++}
&&
\PP(B \cap K)
=
n^{-2}\beta^{-1/2} a(k)
+
n^{-2}b(k)+o(n^{-2}).
\end{eqnarray}

Proof of (\ref{A1+++}) and (\ref{B1+++}). 
Introduce the sets of indices
\begin{eqnarray}
\nonumber
&&
{\cal C}_1=[m],
\qquad
{\cal C}_2=\{(i,j):\, i\not=j;\ i,j\in [m]\},
\\
\nonumber
&&
{\cal C}_3=\{(i,j,k):\, i\not=j\not=k\not=i;\ i,j,k\in [m]\}
\end{eqnarray} 
and split
\begin{eqnarray}\nonumber
&&
B=B_1\cup B_2,
\quad
A=B_1\cup B_3,
\quad
B_k = \bigcup_{x\in {\cal C}_k} B_{k.x},
\quad
k=1,2,3,
\end{eqnarray}
where
\begin{eqnarray}
  &&
 B_{1.i} = \{w_i \to v_1, \, w_i \to v_2, \, w_i \to v_3\}, \nonumber 
  \\
  \nonumber
 &&
  B_{2.(i,j)} = \{w_i \to v_1, \, w_i \to v_2, \, w_j \to v_1, \,  w_j \to v_3\},
 \\
  \nonumber 
&& 
B_{3.(i,j,k)} 
 = 
\{w_i \to v_1, \, w_i \to v_2, \, w_j \to v_1, \,  w_j \to v_3, 
\,
w_k \to v_2, \, w_k \to v_3\}.
\end{eqnarray}
We write
\begin{eqnarray}\label{A+}
&&
\PP(A \cap K)
= 
\PP(B_1 \cap K) 
+
\PP((B_3\cap K) \setminus B_1),
\\
\label{B+} 
&&
\PP(B \cap K)
= 
\PP(B_1 \cap K) + \PP(B_2 \cap K) 
-
\PP(B_1 \cap B_2 \cap K)
 \end{eqnarray}
and evaluate $\PP(B_k\cap K)$, for $k=1,2$,  
using inclusion-exclusion, 
\begin{equation}
\label{B++} 
\sum_{x\in {\cal C}_k}\PP(B_{k.x} \cap K)
-
\sum_{\{x,y\} \subset {\cal C}_k} \PP(B_{k.x}\cap B_{k.y})
\leq 
\PP(B_k\cap K)
\le 
\sum_{x\in {\cal C}_k}\PP(B_{k.x} \cap K). 
\end{equation}
 We show in Lemma \ref{L1} below that the quantities
\begin{eqnarray}
\label{R_i}
&&
R_{k}:=
\sum_{\{x,y\} \subset {\cal C}_k} \PP(B_{k.x}\cap B_{k.y}),
\quad
k=1,2,
\\
\nonumber
&&
R_3:=\PP((B_3\cap K) \setminus B_1),
\qquad
R_4:=
\PP(B_1 \cap B_2 \cap K)
\end{eqnarray}
are negligibly small. More precisely, we establish the bounds  $R_i=O(n^{-3})$, $1\le i\le 4$. Invoking 
these bounds in (\ref{A+}), 
(\ref{B+}), (\ref{B++}) we obtain
 \begin{eqnarray}\label{A+++}
\PP(A \cap K)
&=& 
\PP(B_1\cap K)+o(n^{-2})
=
m\PP(B_{1.1} \cap K) 
+ o(n^{-2}),
\\
\label{B+++} 
\PP(B \cap K)
&=&
\PP(B_1\cap K)
+
\PP(B_2\cap K)
+o(n^{-2})
\\
\nonumber
&=& 
m\PP(B_{1.1} \cap K) + (m)_2\PP(B_{2.(1,2)} \cap K) 
+o(n^{-2}).
 \end{eqnarray}
In the remaining part of the proof we evaluate the probabilities 
\begin{displaymath}
 p_1:=\PP(B_{1.1} \cap K)
\quad
{\rm{and}}
\quad
p_2:=\PP(B_{2.(1,2)} \cap K).
\end{displaymath}
We shall show that 
\begin{equation}\label{21+3}
(nm)^{3/2}p_1=a(k)+o(1)
\qquad
{\rm{and}} 
\qquad
(nm)^2p_2=b(k)+o(1).
\end{equation}
Finally,  invoking (\ref{21+3}) in (\ref{A+++}), (\ref{B+++}) we 
obtain (\ref{A1+++}), (\ref{B1+++}) thus proving (\ref{T1-21-1}).

\smallskip

It remains to prove (\ref{21+3}). For convenience we divide the proof into three steps.
For this part of the proof we need some more notation.
Let $d_1^*$ (respectively $d_2^*$) denote the number of neighbors of $v_1$ 
in $V^*=\{v_4,v_5,\cdots, v_n\}$  witnessed by the attribute $w_1$ 
(respectively $w_2$). Let $d'_1$ (respectively $d'_2$) denote 
the number of neighbors of $v_1$ in $V^*$ witnessed by some attributes from 
$W'_1=\{w_2,w_3,\dots, w_m\}$ (respectively $W'_2=\{w_3,w_4,\dots, w_m\}$). 

{\it Step 1}. We firstly show that
\begin{eqnarray}\label{AA++}
&&
p_1=\PP\bigl(B_{1.1}\cap\{d_1^*+d'_1=k-2\}\bigr)+O(n^{-4}),
\\
\label{BB++}
&&
p_2=\PP\bigl(B_{2.(1,2)} \cap \{d_1^*+d_2^*+d'_2=k-2\}\bigr)+O(n^{-5}).
\end{eqnarray} 

To show (\ref{AA++}) we count neighbors of $v_1$ in $V^*$. The number of such 
neighbors
is denoted $d^*(v_1)$. We have 
$d^*(v_1)=d_1^*+d'_1-d_0$, where
$d_0$  is the number of neighbors of $v_1$ witnessed by $w_1$ and by some 
attribute(s) 
$w_i\in W'_1$ simultaneously.
Combining the inequality 
\begin{displaymath}
d_0 \le \sum_{j=4}^n 
\left( 
  {\mathbb I}_{1j} {\mathbb I}_{11} \sum_{i=2}^m{\mathbb I}_{ij} 
    {\mathbb I}_{i1} 
\right)
\end{displaymath}
with  Markov's inequality we obtain
\begin{displaymath}
\PP\bigl(B_{1.1}\cap\{d_0\ge 1\})
\le
\E {\mathbb I}_{B_{1.1}}d_0
\le
(n-3)(m-1)\E{\mathbb I}_{B_{1.1}}{\mathbb I}_{14}{\mathbb I}_{11}
  {\mathbb I}_{24} {\mathbb I}_{21}. 
\end{displaymath}
Furthermore, invoking the inequality
\begin{displaymath}
 \E{\mathbb I}_{B_{1.1}}{\mathbb I}_{14}{\mathbb I}_{11}{\mathbb I}_{24}
 {\mathbb I}_{21}
=
\E p_{11} p_{12} p_{13} p_{14} p_{21} p_{24}
\le a_2a_4b_1^2b_2^2(nm)^{-3}
\end{displaymath}
we obtain 
$\PP\bigl(B_{1.1}\cap\{d_0\ge 1\})=O(n^{-4})$. Now (\ref{AA++})
 follows from the fact that the event  $B_{1.1}$ implies  $d(v_1)=d^*(v_1)+2$.

The proof of  (\ref{BB++}) is almost the same.
We color 
$w_1$ red, $w_2$ green and all $w_i\in W_2'$ we color yellow.
Let $d'_0$  denote the number of neighbors of $v_1$ witnessed 
by at least two attributes of different colors.
Note that the
number $d^*(v_1)$ of
neighbors of $v_1$ in $V^*$
satisfies, by inclusion-exclusion,
\begin{equation}\label{dddd}
d_1^*+d_2^*+d'_2-2d_0'\le d^*(v_1)\le d_1^*+d_2^*+d'_2.
\end{equation}
We combine the inequality
\begin{displaymath}
d_0'
\le 
\sum_{j=4}^n
\left(
{\mathbb I}_{11}{\mathbb I}_{1j}{\mathbb I}_{21}{\mathbb I}_{2j}
+
({\mathbb I}_{11}{\mathbb I}_{1j}+{\mathbb I}_{21}{\mathbb I}_{2j})
\sum_{i=3}^m{\mathbb I}_{i1}{\mathbb I}_{ij}
\right)
\end{displaymath}
with the identity ${\mathbb I}_{B_{2.(1,2)}}{\mathbb I}_{11}{\mathbb I}_{21}={\mathbb I}_{B_{2.(1,2)}}$ and obtain, by Markov's inequality
and symmetry, that
\begin{displaymath}
\PP\bigl(B_{2.(1,2)}\cap\{d_0'\ge 1\})
\le
\E {\mathbb I}_{B_{2.(1,2)}}d_0'
\le
(n-3)\E{\mathbb I}_{B_{2.(1,2)}}
\Bigl(
{\mathbb I}_{14}{\mathbb I}_{24} 
+
2(m-2){\mathbb I}_{14}{\mathbb I}_{31}{\mathbb I}_{34}
\Bigr). 
\end{displaymath}
Furthermore, invoking the inequalities
\begin{eqnarray}\nonumber
&&
\E {\mathbb I}_{B_{2.(1,2)}}{\mathbb I}_{14}{\mathbb I}_{24}
=
\E p_{11}p_{12}p_{14}p_{21}p_{23}p_{24}
\le
a_3^2b_1^2b_2^2(mn)^{-3},
\\
\nonumber
&&
 \E
 {\mathbb I}_{B_{2.(1,2)}}
 {\mathbb I}_{14}
 {\mathbb I}_{31}
 {\mathbb I}_{34}
=
\E p_{11}p_{12}p_{14}p_{21}p_{23}p_{31}p_{34}
\le
a_2^2a_3b_1^2b_2b_3(mn)^{-7/2}
\end{eqnarray}
we obtain
$\PP\bigl(B_{2.(1,2)}\cap\{d'_0\ge 1\})=O(n^{-5})$. 
Now (\ref{BB++}) follows from 
(\ref{dddd}) and the identity $d(v_1)=d^*(v_1)+2$.


{\it Step 2}. We secondly show that
\begin{eqnarray}\label{LL20-5}
&&
(nm)^{3/2}p_1
=
b_1^2
\E \Bigl(X_1^3Y_1\PP\bigl(\Lambda_1+d_*=k-2\, \bigr|\, 
X_1,Y_1\bigr)\Bigr)
+
o(1),
\\
\label{LL20-6}
&&
(nm)^2p_2
=
b_1^2
\E\Bigl( X_1^2X_2^2Y_1\PP\bigl(\Lambda_1+\Lambda_2+d_*=k-2\, 
\bigr|\, X_1,X_2,Y_1\bigr)\Bigr)
+o(1).
\
\end{eqnarray}

 Let us prove (\ref{LL20-5}). We have 
\begin{eqnarray}
\label{LL20-3}
\PP\bigl(B_{1.1}\cap\{d_1^*+d'_1=k-2\}\bigr)
&
=
&
\E \bigl(p_{11}p_{12}p_{13}\PP^*(d_1^*+d_1'=k-2)\bigr)
\\
\nonumber
&
=
&
\E \bigl( \lambda_{11}\lambda_{12}\lambda_{13}\PP^*(d_1^*+d_1'=k-2)\bigr)
+o((nm)^{-3/2})
\end{eqnarray}
and
\begin{eqnarray}\label{LL20-4}
(nm)^{3/2}\E \bigl( \lambda_{11}\lambda_{12}\lambda_{13}\PP^*(d_1^*+d_1'=k-2)\bigr)
&
=
&
b_1^2\E\bigl( X_1^3Y_1\PP^*(d_1^*+d_1'=k-2)\bigr)
\\
\nonumber
&
=
&
b_1^2\E\bigl( X_1^3Y_1\PP^*(d_*+\Lambda_1=k-2)\bigr)+o(1).
\end{eqnarray}
Here (\ref{LL20-4}) follows from Lemma \ref{LL}, by Lebesgue's dominated convergence theorem. Furthermore, (\ref{LL20-3}) follows from the inequalities
\begin{displaymath}
\lambda_{11}\lambda_{12}\lambda_{13}
\ge 
p_{11}p_{12}p_{13}
\ge 
\lambda_{11}\lambda_{12}\lambda_{13}
\bigl(
1-{\mathbb I}_{\{\lambda_{11}>1\}}
-{\mathbb I}_{\{\lambda_{12}>1\}}
-{\mathbb I}_{\{\lambda_{13}>1\}}
\bigr)
\end{displaymath}
combined with the simple bound
\begin{displaymath}
\E \Bigr(\lambda_{11}\lambda_{12}\lambda_{13}
(
{\mathbb I}_{\{\lambda_{11}>1\}}
+{\mathbb I}_{\{\lambda_{12}>1\}}
+{\mathbb I}_{\{\lambda_{13}>1\}}
)
\Bigr)
=o\Bigr( \E \bigr(\lambda_{11}\lambda_{12}\lambda_{13}\bigr)\Bigr)
=
o\bigr((nm)^{-3/2}\bigr).
\end{displaymath}
Note that (\ref{AA++}), (\ref{LL20-3}), (\ref{LL20-4}) imply 
(\ref{LL20-5}).

The proof of (\ref{LL20-6}) is much the same.
We have 
\begin{eqnarray}\nonumber
&&
\PP\bigl(B_{2.(1,2)} \cap \{d_1^*+d_2^*+d'_2=k-2\}\bigr)
\\
\nonumber
&&
=
\E \bigl(p_{11}p_{12}p_{21}p_{23}
\PP^{**}(d_1^*+d_2^*+d'_2=k-2)\bigr)
\\
\nonumber
&&
=
\E \bigl( \lambda_{11}\lambda_{12}\lambda_{21}\lambda_{23}
\PP^{**}(d_1^*+d_1'=k-2)\bigr)
+o((nm)^{-2})
\end{eqnarray}
and
\begin{eqnarray}\nonumber
&&
(nm)^2
\E \bigl( \lambda_{11}\lambda_{12}\lambda_{21}\lambda_{23}
\PP^{**}(d_1^*+d_2^*+d'_2=k-2)\bigr)
\\
\nonumber
&&
=
b_1^2
\E\bigl( X_1^2X_2^2Y_1^2\PP^{**}(d_1^*+d_2^*+d'_2=k-2)\bigr)
\\
\nonumber
&&
=
b_1^2
\E\bigl( X_1^2X_2^2Y_1^2
\PP^{**}(d_*+\Lambda_1+\Lambda_2=k-2)\bigr)+o(1).
\end{eqnarray}

{\it Step 3}. In this final step we show that 
\begin{eqnarray}\label{LL20-7}
&&
\E \bigl(X_1^3Y_1\PP^*(d_*+\Lambda_1=k-2)\bigr)
=
a_3b_1\PP(d_*^{(1)}+\Lambda_1^{(3)}=k-2).
\\
\label{LL20-8}
&&
\E\bigl( X_1^2X_2^2Y_1\PP^{**}(d_*+\Lambda_1+\Lambda_2=k-2)\bigr)
\\
\nonumber
&&
=
a_2^2b_2\PP(d_*^{(2)}+\Lambda_1^{(2)}+\Lambda_2^{(2)}=k-2).
\end{eqnarray}

In the proof we use the observation that
\begin{eqnarray}\nonumber
\E
\Bigl(
Y_1^r
\PP_{Y_1}(d_*=s)
\Bigr)
&
=
&
\E
\sum_{i\ge 0}
\left(
Y_1^r
\PP_{Y_1}(\Lambda_0=i)
\PP\left(\sum_{j=0}^i\tau_j=s\right)
\right)
\\
\nonumber
&
=
&
\sum_{i\ge 0}
\left( b_r\PP\bigl(\Lambda_0^{(r)}=i\bigr)
\PP\left(\sum_{j=0}^i\tau_j=s\right)
\right)
\\
\nonumber
&
=
&
b_r\PP\bigl(d_*^{(r)}=s\bigr).
\end{eqnarray}

To show (\ref{LL20-7}) we write the quantity on the left in the form
\begin{eqnarray}\nonumber
&&
\E
\left(
X_1^3Y_1
\sum_{s+t=k-2}
\PP_{Y_1}\bigl(d_*=s\bigr)
\cdot
\PP_{X_1}(\Lambda_1=t)
\right)
\\
\nonumber
&&
=
\sum_{s+t=k-2}b_1\PP\bigl(d_*^{(1)}=s\bigr)\cdot a_3\PP\bigl(\Lambda_1^{(3)}=t\bigr)
\\
\nonumber
&&
=
b_1a_3\PP\bigl(d_*^{(1)}+\Lambda_1^{(3)}=k-2\bigr).
\end{eqnarray}

To show (\ref{LL20-8}) we write the quantity on the left in the form
\begin{eqnarray}\nonumber
&&
\E
\left(
X_1^2X_2^2Y_1^2
\sum_{s+t+u=k-2}
\PP_{Y_1}\bigl(d_*=s\bigr)
\cdot
\PP_{X_1}(\Lambda_1=t)
\cdot
\PP_{X_2}(\Lambda_2=u)
\right)
\\
\nonumber
\quad
&&
=
\sum_{s+t+u=k-2}b_2\PP\bigl(d_*^{(2)}=s\bigr)
\cdot 
a_2\PP\bigl(\Lambda_1^{(2)}=t\bigr)
\cdot
a_2
\PP\bigl(\Lambda_2^{(2)}=u\bigr)
\\
\nonumber
&&
=
\
b_2a_2^2
\PP\Bigl(d_*^{(2)}+\Lambda_1^{(2)}+\Lambda_2^{(2)}=k-2\Bigr).
\end{eqnarray}
\qed

{\it Proof of Theorem \ref{T0}.} 
In the proof we use shorthand notation
 ${\tilde A}(k)=\PP(d_*^{(1)}+\Lambda_1^{(3)}\ge k)$
and 
${\tilde B}(k)=\PP(d_*^{(2)}+\Lambda_1^{(2)}+\Lambda_2^{(2)}\ge k)$.
Given two positive functions $f(t)$ and $g(t)$ we denote $f(t)\simeq g(t)$ whenever
$f(t)/g(t)\to 1$ as $t\to+\infty$. 

Using asymptotic formulas for the tail probabilities of randomly stopped sums 
$d_*^{(r)}$ reported in \cite{AleskevicieneLeipusSiaulys}, 
and the formulas for the tail probabilities 
of $\Lambda_k^{(r)}$ shown in Lemma \ref{LB1}, we obtain 
\begin{eqnarray}
\nonumber 
\PP(d_*^{(1)}\ge k)
&
\simeq 
&
c_Y
\frac{\gamma}{\gamma-1}
a_2^{\gamma-1}
b_1^{\gamma-2}
k^{1-\gamma},
\\
\label{dKetv1}
\PP(d_*^{(2)}\ge k)
&
\simeq
&
c_Y
\frac{\gamma}{\gamma-2}
a_2^{\gamma-2}
b_1^{\gamma-2}
b_2^{-1}
k^{2-\gamma},
\\
\nonumber
\PP(\Lambda_1^{(r)}\ge k)
&
\simeq 
&
c_X
\frac{\alpha}{\alpha-r}
\beta^{(r-\alpha)/2}
a_r^{-1}
b_1^{\alpha-r}
k^{r-\alpha},
\qquad
r=2,3.
\end{eqnarray}
Next we combine these asymptotic formulas with the aid of Lemma \ref{LB2}.  
We have
\begin{eqnarray}\nonumber
&& {\tilde A}(k)\simeq
\PP(d_*^{(1)}\ge k),
\qquad
\qquad
\qquad
\quad
\
{\rm{for}}
\quad  
\alpha>\gamma+2,
\\
\label{dKetv2}
&& {\tilde A}(k)\simeq
\PP(d_*^{(1)}\ge k)+\PP(\Lambda_1^{(3)}\ge k),
\quad
{\rm{for}}
\quad  
\alpha=\gamma+2,
\\
\nonumber
&& {\tilde A}(k)\simeq
\PP(\Lambda_1^{(3)}\ge k),
\quad
\qquad
\qquad
\qquad
\
{\rm{for}}
\quad  
\alpha<\gamma+2
\end{eqnarray}

and
\begin{eqnarray}\nonumber
&& {\tilde B}(k)\simeq
\PP(d_*^{(2)}\ge k),
\qquad
\qquad
\qquad
\qquad
\qquad
\qquad
\quad
\
\
{\rm{for}}
\quad  
\alpha>\gamma,
\\
\label{dKetv3}
&& {\tilde B}(k)\simeq
\PP(d_*^{(2)}\ge k)+\PP(\Lambda_1^{(2)}\ge k)+\PP(\Lambda_2^{(2)}\ge k),
\quad
{\rm{for}}
\quad  
\alpha=\gamma,
\\
\nonumber
&& {\tilde B}(k)\simeq
\PP(\Lambda_1^{(2)}\ge k)+\PP(\Lambda_2^{(2)}\ge k),
\quad
\qquad
\qquad
\qquad
\
{\rm{for}}
\quad  
\alpha<\gamma.
\end{eqnarray}
Finally, from (\ref{dKetv1}), (\ref{dKetv2}), (\ref{dKetv3})
we
derive
(\ref{T0+}).
\qed


\section{Auxiliary lemmas}

Let ${\tilde d}_1^*$ (respectively ${\tilde d}_2^*$) denote the number vertices 
in $V^*=\{v_4,v_5,\cdots, v_n\}$ linked to   the attribute $w_1$ 
(respectively $w_2$). 
Let $x_1,x_2, y_1\ge 0$. For $k=1,2$, 
let ${\tilde d}_k$, ${\tilde \Lambda}_k$ denote the random variables 
${\tilde d}_k^*$, $\Lambda_k$ conditioned on the event $X_k=x_k$ (to 
get ${\tilde d}_k$, ${\tilde \Lambda}_k$  we replace $X_k$ by a non-
random number $x_k$ in the definition of 
${\tilde d}_k^*$, $\Lambda_k$).
Let ${\hat d}_1$, ${\hat d}_2$ and ${\hat d}_*$ denote the random 
variables $d'_1$,  $d'_2$ and $d_*$ conditioned on the event $Y_1=y_1$ 
(to get ${\hat d}_1$, ${\hat d}_2$ and ${\hat d}_*$  we replace $Y_1$ 
by a non-random number $y_1$ in the definition of $d'_1$,  $d'_2$ 
and $d_*$).

\begin{lem}\label{LL} Let $\beta>0$. Let $n,m\to+\infty$. Assume that $m/n\to\beta$. Assume that $\E X_i^2<\infty$ and $\E Y_j<\infty$. 
For any $x_1,x_2,y_1\ge 0$  and $s,t,u =0,1,2,\dots$, we have
\begin{equation}\label{LL1}
\PP({\hat d}_1=s, \, {\tilde d}_1=t)
\to
\PP({\hat d}_*=s,\, {\tilde \Lambda}_1=t)
=
\PP({\hat d}_*=s)\PP({\tilde \Lambda}_1=t),
\end{equation}
\begin{eqnarray}\nonumber
\PP({\hat d}_2=s, \, {\tilde d}_1=t, \, {\tilde d}_2=u)
&
\to
& 
\PP({\hat d}_*=s,\, {\tilde \Lambda}_1=t,\,{\tilde \Lambda}_2=u)
\\
\label{LL2}
&
=
&
\PP({\hat d}_*=s)\PP({\tilde \Lambda}_1=t)\PP({\tilde \Lambda}_2=u).
\end{eqnarray}
\end{lem}

We remark that (\ref{LL1}) tells us that  random vector
$({\hat d}_1, {\tilde d}_1)$ converges in distribution to the random 
vector 
$({\hat d}_*, {\tilde \Lambda}_1)$. 
Similarly,  (\ref{LL2}) tells us that  random vector
$({\hat d}_1, {\tilde d}_1, {\tilde d}_2)$ converges in distribution 
to the random vector 
$({\hat d}_*, {\tilde \Lambda}_1, {\tilde \Lambda}_2)$. In particular, 
(\ref{LL1})  implies for any $r=0,1,2,\dots$ that 
\begin{displaymath}
\PP({\hat d}_1+ {\tilde d}_1\ge r)\to
\PP({\hat d}_*+ {\tilde \Lambda}_1\ge r)
\qquad
{\rm{and}}
\qquad
\PP({\hat d}_1+ {\tilde d}_1= r)\to
\PP({\hat d}_*+ {\tilde \Lambda}_1= r)
\end{displaymath} 
as 
$n,m\to+\infty$.
(\ref{LL2})
implies that 
\begin{eqnarray}\nonumber
&&
\PP({\hat d}_2+ {\tilde d}_1+{\tilde d}_2\ge r)\to
\PP({\hat d}_*+ {\tilde \Lambda}_1+{\tilde \Lambda}_2\ge r),
\\
\nonumber
&&
\PP({\hat d}_2+ {\tilde d}_1+{\tilde d}_2= r)\to
\PP({\hat d}_*+ {\tilde \Lambda}_1+{\tilde \Lambda}_2= r.
\end{eqnarray}

{\it Proof of Lemma \ref{LL}.} 
 Before the proof we introduce some notation.
Let $\PP_1$ denote the conditional probability given $\{Y_4, Y_5,\dots, Y_n\}$. 
For $a>0$ and $s=0,1,2\dots$ we denote by
$f_s(a)=a^se^{-a}/s!$ the Poisson probability. Below we use the fact that $|f_s(a)-f_s(b)|\le |a-b|$. Furthermore we denote
\begin{displaymath}
{\tilde \lambda}_k=x_k\beta^{-1/2}b_1
\qquad
{\rm{and}}
\qquad
{\tilde \lambda}_{3|k}=\sum_{j=4}^n{\tilde \lambda}_{kj},
\quad
{\tilde \lambda}_{4|k}=\sum_{j=4}^n{\tilde p}_{kj},
\qquad
k=1,2. 
\end{displaymath}
Here ${\tilde p}_{kj}$, ${\tilde \lambda}_{kj}$
are defined in the same way as $p_{kj}$, $\lambda_{kj}$, but with $X_k$ replaced by $x_k$, for $k=1,2$.

Proof of (\ref{LL1}).  We have
\begin{equation}\label{LL18-0}
\PP({\hat d}_1=s, \, {\tilde d}_1=t)
=
\E \PP_1({\hat d}_1=s, \, {\tilde d}_1=t)
=
\E \bigl(\PP_1({\hat d}_1=s)\PP_1({\tilde d}_1=t)\bigr).
\end{equation}
Given $\{Y_4, Y_5,\dots, Y_n\}$, the random variable ${\tilde d}_1$ is a sum of independent Bernoulli random variables. 
We invoke Le Cam's inequality, see, e.g, \cite{Steele},
\begin{equation}\label{LL18-1}
\bigl| \PP_1({\tilde d}_1=t)-f_t({\tilde \lambda}_{4|1})|
\le 
\sum_{j=4}^n{\tilde p}_{1j}^2=:R_1^*
\end{equation}
and use simple inequalities
\begin{eqnarray}\label{LL18-2}
&&
\bigl|
f_t({\tilde \lambda}_{4|1})-f_t({\tilde \lambda}_{3|1})
\bigr|
\le 
|{\tilde \lambda}_{4|1}-{\tilde \lambda}_{3|1}|
\le 
\sum_{j=4}^n{\tilde \lambda}_{1j}
{\mathbb I}_{\{{\tilde \lambda}_{1j}>1\}}
=:R_2^*,
\\
\label{LL18-3}
&&
\bigl|f_t({\tilde \lambda}_{3|1})-f_t({\tilde \lambda}_1)\bigr|
\le 
|{\tilde \lambda}_{3|1}-{\tilde \lambda}_1|
=
x_1\left|\sqrt{n/m}\left(n^{-1}\sum_{j=4}^nY_j\right)-\beta^{-1/2}b_1
\right|.
\
\
\
\
\
\end{eqnarray}
Note that 
$\bigl|f_t({\tilde \lambda}_{3|1})-f_t({\tilde \lambda}_1)\bigr|
\to 0$ almost surely, by the law of large numbers. 
Furthermore, 
\begin{displaymath}
\E R_2^*
=
(n-4)(nm)^{-1/2}x_1\E Y_4{\mathbb I}_{\{x_1Y_4>\sqrt{nm}\}}
=
o(1),
\end{displaymath}
because $\E Y_4{\mathbb I}_{\{x_1Y_4>\sqrt{nm}\}}=o(1)$.
We similarly show that $\E R_1^*=o(1)$. For any 
$\varepsilon\in (0,1)$ 
the inequality 
${\tilde p}_{1j}^2
\le
 {\tilde \lambda}_{1j}
 \Bigl(\varepsilon+{\mathbb I}_{\{{\tilde \lambda}_{1j}>\varepsilon\}}\Bigr)$ implies 
\begin{displaymath}
\E R_1^*=(n-4)\E {\tilde p}_{i4}^2
\le 
(n-4)\E {\tilde \lambda}_{14}
 \Bigl(\varepsilon+{\mathbb I}_{\{{\tilde \lambda}_{14}>\varepsilon\}}\Bigr)
\le 
(n-4)(nm)^{-1/2}\bigl(x_1b_1\varepsilon +o(1)\bigr).
\end{displaymath}
We obtain the bound  $\E R_1^*\le \beta^{-1/2}x_1b_1\varepsilon+o(1)$, 
which implies $\E R_1^*=o(1)$.

Now it follows from (\ref{LL18-1}), (\ref{LL18-2}), (\ref{LL18-3})
that
\begin{eqnarray}\label{LL18-4}
 \E \bigl(\PP_1({\hat d}_1=s)\PP_1({\tilde d}_1=t)\bigr)
 &
 =
 &
 \E \bigl(\PP_1({\hat d}_1=s)f_t({\tilde \lambda}_1)\bigr)+o(1)
\\
\nonumber
&
=
&
 \PP({\hat d}_1=s)f_t({\tilde \lambda}_1)+o(1).
\end{eqnarray}
Next we use the fact that $\PP({\hat d}_1=s)\to \PP({\hat d}_*=s)$. The proof of this fact repeats literally 
the proof of statement (ii) of Theorem 1 of \cite{BloznelisDamarackas2013}. 
Finally, from (\ref{LL18-0}), (\ref{LL18-4}) we obtain (\ref{LL1}): 
\begin{displaymath}
\PP({\hat d}_1=s, \, {\tilde d}_1=t)
= 
 \E \bigl(\PP_1({\hat d}_1=s)\PP_1({\tilde d}_1=t)\bigr)
 \to 
 \PP({\hat d}_*=s)f_t({\tilde \lambda}_1).
\end{displaymath}

Proof of (\ref{LL2}). It is similar to that of (\ref{LL1}).
We have
\begin{equation}\label{LL20-0}
\PP
\bigl(
{\hat d}_2=s, \, {\tilde d}_1=t,\,{\tilde d}_2=u 
\bigr)
=
\E \bigl(
\PP_1({\hat d}_2=s)\PP_1({\tilde d}_1=t)\PP_1({\tilde d}_2=u)\bigr).
\end{equation}
By the same argument as above (see  (\ref{LL18-1}), (\ref{LL18-2}), (\ref{LL18-3})), we obtain  
\begin{eqnarray}\nonumber
\E \bigl(
\PP_1({\hat d}_2=s)\PP_1({\tilde d}_1=t)\PP_1({\tilde d}_2=u)\bigr)
&
=
&
\E \bigl(
\PP_1({\hat d}_2=s)f_t({\tilde \lambda}_1)\PP_1({\tilde d}_2=u)\bigr) +o(1)
\\
\nonumber
&
=
&
\E \bigl(
\PP_1({\hat d}_2=s)f_t({\tilde \lambda}_1)f_u({\tilde \lambda}_2)\bigr) +o(1)
\\
\label{LL20-1}
&
=
&
f_t({\tilde \lambda}_1)f_u({\tilde \lambda}_2)
\PP({\hat d}_2=s)+o(1).
\end{eqnarray}
Finally, 
we use the fact that $\PP({\hat d}_2=s)\to \PP({\hat d}_*=s)$. The proof of this fact repeats literally 
the proof of  statement (ii) of Theorem 1 of \cite{BloznelisDamarackas2013}. 
Now  from (\ref{LL20-0}), (\ref{LL20-1}) we obtain (\ref{LL2}). 
\qed

\begin{lem}\label{L1} The quantities 
$R_i$, $1\le i\le 4$ defined in (\ref{R_i}) 
satisfy $R_i=O(n^{-3})$. 
\end{lem}

{\it Proof of Lemma \ref{L1}.}
The bound $R_1=O(n^{-3})$ is obtained from the identity
$R_1 ={{m}\choose{2}}  \PP(B_{1.1} \cap B_{1.2})$ and inequalities
\begin{eqnarray}\nonumber
\PP(B_{1.1} \cap B_{1.2}) 
&
=
& \E {\mathbb I}_{11} {\mathbb I}_{12} {\mathbb I}_{13}
    {\mathbb I}_{21} {\mathbb I}_{22} {\mathbb I}_{23}    
\\
\nonumber
&
\le
&
\E \lambda_{11} \lambda_{12} \lambda_{13} 
    \lambda_{21} \lambda_{22} \lambda_{23}
  = a_{3}^2 b_{2}^3(nm)^{-3}. 
  \end{eqnarray}
The bound $R_2=O(n^{-3})$ follows from inequalities
\begin{eqnarray}\nonumber
R_2
&=&
\sum_{\{(i,j),(i,r)\}\subset{\cal C}_2, j\not=r} \PP(B_{2.(i,j)} \cap B_{2.(i,r)})
\\
\nonumber
&
+
& 
\sum_{\{(i,j),(k,j)\}\subset{\cal C}_2, i\not=k} \PP(B_{2.(i,j)} \cap B_{2.(k,j)})
\\
\nonumber
&+& 
\sum_{\{(i,j),(j,i)\}\subset{\cal C}_2} \PP(B_{2.(i,j)} \cap B_{2.(j,i)})
\\
\nonumber
&
+
&
 \sum_{\{(i,j),(k,r)\}\subset{\cal C}_2, i \neq j \neq k \neq r\not=i} 
 \PP(B_{2.(i,j)} \cap B_{2.(k,r)}) 
 \\
 \nonumber
&=&
2^{-1} (m)_3 \PP(B_{2.(1,2)} \cap B_{2.(1,3)}) 
+ 
2^{-1}(m)_3 \PP(B_{2.(1,3)} \cap B_{2.(2,3)}) 
\\
\nonumber
&+&
(m)_2 \PP(B_{2.(1,2)} \cap B_{2.(2,1)}) 
+
2^{-1}(m)_4 \PP(B_{2.(1,2)} \cap B_{2.(3,4)}) 
\\
\nonumber
&=& 
2^{-1}(m)_3 \E p_{11} p_{12} p_{21} p_{23} p_{31} p_{33}
+ 
2^{-1}(m)_3 \E p_{11} p_{12} p_{21} p_{22} p_{31} p_{33} 
\\
\nonumber
&+&
(m)_2 \E p_{11} p_{12} p_{13} p_{21} p_{22} p_{23}
+
2^{-1}(m)_4 \E p_{11} p_{12} p_{21} p_{23} p_{31} p_{32} p_{41} p_{43}
\\
\nonumber
&
\leq
&  (m)_3 \frac{a_{2}^3 b_{1} b_{2} b_{3}}{(nm)^3} 
    + (m)_2 \frac{a_{3}^2 b_{2}^3}{(nm)^3}
    + 2^{-1}(m)_4 \frac{a_{2}^4 b_{2}^2 b_{4}}{(nm)^4}. 
\end{eqnarray}
The bound $R_3=O(n^{-3})$ is obtained  from the inequalities
\begin{eqnarray}\nonumber
\PP(B_3)
&
\le
& 
\sum_{x\in{\cal C}_3} \PP(B_{3.x}) 
=
(m)_3 \PP(B_{3.(1,2,3)}) 
\\
\nonumber
&
=
&
(m)_3 \E p_{11} p_{12} p_{21} p_{23} p_{32} p_{33}
\le 
\frac{(m)_3}{(nm)^3}a_{2}^3 b_{2}^3.
\end{eqnarray}
The bound $R_4=O(n^{-3})$ is obtained  from the inequalities
\begin{eqnarray}\label{B1-B2-1}
\PP(B_1 \cap B_2)  
&
\le
& 
\sum_{y\in{\cal C}_2} \PP(B_1 \cap B_{2.y}) 
= 
 m(m - 1) \PP(B_1 \cap B_{2.(1,2)} ),
\\
\nonumber
\PP(B_{1} \cap B_{2.(1,2)})
&
\le
& 
\PP(B_{1.1} \cap B_{2.(1,2)}) + \PP(B_{1.2} \cap B_{2.(1,2)})
\\
\nonumber
&
+
& 
(m-2) \PP(B_{1.3} \cap B_{2.(1,2)})
\end{eqnarray}
and bounds
\begin{eqnarray}
\nonumber
&&
\PP(B_{1.1} \cap B_{2.(1,2)})
= 
\PP(B_{1.2} \cap B_{2.(1,2)})
= 
\E p_{11} p_{12} p_{13} p_{21} p_{23}
\le
a_{2} a_{3} b_{1} b_{2}^2(nm)^{-5/2},
\\
\nonumber
&&
\PP(B_{1.3} \cap B_{2.(1,2)}) 
= 
\E p_{11} p_{12} p_{21} p_{23} p_{31} p_{32} p_{33}
\le
a_{2}^2 a_{3} b_{2}^2 b_{3}(nm)^{-7/2}. 
\end{eqnarray}
\qed

\begin{lem}\label{LB1} Let $\alpha, c>0$. Let $r$ be an integer and $0\le r<\alpha$. Let $t\to+\infty$. For a non-negative random variable $Z$ satisfying $\PP(Z>t)=(c+o(1))t^{-\alpha}$  we have 
\begin{equation}\label{2017-03-28}
\E \bigl( Z^r{\mathbb I}_{\{Z>t\}}\bigr)
=
(c+o(1))\alpha(\alpha-r)^{-1}t^{r-\alpha}.
\end{equation}
Denote $h_r=\E Z^r$. For a random variable $\Lambda_Z$ with the distribution $\PP(\Lambda^{(r)}_Z=k)=h^{-1}_r\E \bigl(e^{-Z}Z^{k+r}/k!\bigr)$, $k=0,1,2,\dots$, we have
\begin{equation}\label{2017-03_28+}
\PP(\Lambda_Z^{(r)}>t)
=
(1+o(1))h_r^{-1}\E \bigl(Z^r{\mathbb I}_{\{Z>t\}}\bigr)
=
(1+o(1))h_r^{-1}c\alpha(\alpha-r)^{-1}t^{r-\alpha}.
\end{equation}
\end{lem}

{\it Proof of Lemma \ref{LB1}.} 
Denote $F(x)=\PP(Z\le x)=1-{\bar{F}}(x)$.
To show (\ref{2017-03-28}) for $r=1,2,\dots$ we apply integration by parts formula for the Lebesgue-Stieltjes integral 
\begin{eqnarray}
\nonumber
\E (Z^r{\mathbb I}_{\{Z> t\}}\bigr)
&=&
\int_t^{+\infty}x^rdF(x)
=
-
\int_t^{+\infty}x^rd{\bar {F}}(x)
\\
\nonumber
&
=
&
t^r{\bar F}(t)+\int_{t}^{+\infty}rx^{r-1}{\bar F}(x)dx
\end{eqnarray}
and invoke ${\bar{F}}(x) =\PP(Z>x)=(c+o(1))x^{-\alpha}$.

Proof of (\ref{2017-03_28+}). Fix $r$. 
For $s,t,x>0$ and $k=0,1,2,\dots$  we denote
\begin{displaymath}
S^{(k)}_x(s):=\sum_{i<s}e^{-x}x^{i+k}/i!,
\qquad
{\bar S}^{(k)}_x(t):=\sum_{i\ge t}e^{-x}x^{i+k}/i!
\end{displaymath}
For $0<s<x<t$ we will use the  inequalities 
(see \cite{mitzenmacherANDupfal})
\begin{equation}\label{2017-03-28++}
S^{(0)}_x(s)
\le 
e^{s-x}(x/s)^s
\qquad
{\rm{and}}
\qquad
{\bar S}^{(0)}_x(t)
\le 
 e^{t-x}(x/t)^t.
\end{equation}
Given $0<\varepsilon<1$ we write for short $t_{1}=t(1-\varepsilon)$, $t_{2}=t(1+\varepsilon)$ and split the probability
\begin{eqnarray}
\nonumber
&&
\PP(\Lambda_Z^{(r)}>t)
=
h_r^{-1}\E {\bar S}^{(r)}_Z(t)
=
h_r^{-1}(I_1+I_2+I_3),
\qquad
I_k=\E {\bar S}^{(r)}_Z(t){\mathbb I}_{\{Z\in A_k\}},
\\
\nonumber
&&
A_1=[0, t_{1}),
\quad 
A_2=[t_{1}, t_{2}],
\quad
A_3=(t_{2},+\infty).
\end{eqnarray}
We let $\varepsilon=t^{-1/3}$ and evaluate $I_1$, $I_2$ and $I_3$.  The second inequality of (\ref{2017-03-28++}) implies
\begin{equation}\label{2017-03-28+++}
I_1
= 
\E 
\bigl(
Z^r{\bar S}^{(0)}_Z(t){\mathbb I}_{\{Z<t_{1}\}}
\bigr)
\le 
\E
\bigl(
e^{-Z}Z^{t+r}(e/t)^t{\mathbb I}_{\{Z<t_{1}\}}
\bigr)
\le 
e^{-t_{1}}t_{1}^{t+r}(e/t)^t.
\end{equation}
In the last step we used the fact that $z\to e^{-z}z^{t+r}$ is an increasing function on $(0,t_1)$. Furthermore, the quantity on the right of  (\ref{2017-03-28+++}) is less than
\begin{displaymath}
t^re^{t-t_1}(t_1/t)^t
=
t^r e^{\varepsilon t}(1-\varepsilon)^t
=
t^r
e^{t\varepsilon+t\ln(1-\varepsilon)}
\le 
t^re^{-t\varepsilon^2/2}=o(t^{r-\alpha}).
\end{displaymath}
Hence $I_1=o(t^{r-\alpha})$.
While estimating $I_2$ we use the inequalities 
 $t_1^{-\alpha}-t_2^{-\alpha}\le c'\alpha\varepsilon t^{-\alpha-1}$ and  ${\bar S}^{(0)}_x(t)\le 1$. We obtain
\begin{displaymath}
I_2
\le 
\E Z^r{\mathbb I}_{\{t_1\le Z\le t_2\}}
\le
t_2^r
\PP(t_1\le Z\le t_2)
=
t_2^r(t_2^{-\alpha}-t_1^{-\alpha})c(1+o(1))
=
o(t^{r-\alpha}).
\end{displaymath}
We finally evaluate $I_3$. From the identity 
$ S_x^{(0)}(t)+ {\bar S}_x^{(0)}(t)=1$ we obtain 
${\bar S}^{(r)}_x(t)=x^r\bigl(1-S_x^{(0)}(t)\bigr)$. Using this expression we write $I_3$ in the form
\begin{displaymath}
I_3=\E \bigl(Z^r{\mathbb I}_{\{Z>t_2\}}\bigr)+R,
\qquad
{\rm{where}}
\qquad
R=\E\bigl(Z^rS_Z^{(0)}(t){\mathbb I}_{\{Z>t_2\}}\bigr).
\end{displaymath}
Note that (\ref{2017-03-28}) implies
\begin{displaymath}
\E\bigl(Z^r{\mathbb I}_{\{Z>t_2\}}\bigr)
=
(c+o(1))\alpha(\alpha-r)^{-1}t^{r-\alpha}.
\end{displaymath}
We complete the proof by showing that $R=o(t^{r-\alpha})$.
The first inequality of (\ref{2017-03-28++}) implies
\begin{eqnarray}\nonumber
R
&
\le 
&
\E \bigl(Z^{t+r}e^{-Z}(e/t)^t{\mathbb I}_{\{Z>t_2\}}\bigr)
\le t_2^{t+r}e^{-t_2}(e/t)^t
=
t_2^re^{-\varepsilon t}(1+\varepsilon)^t
\\
\nonumber
&
\le
&
t_2^re^{-\varepsilon^2t/4}=o(t^{r-\alpha}).
\end{eqnarray}
In the second inequality we used the fact that the function 
$z\to z^{t+r}e^{-z}$ decreases on $(t_2,+\infty)$. In the last inequality we estimated $\ln (1+\varepsilon)^t=t\ln(1+\varepsilon)\le t(\varepsilon-\varepsilon^2/4)$.
\qed

\noindent In the next lemma we collect several simple facts used in the proof of Theorem~\ref{T0}.

\begin{lem}\label{LB2} Let $\alpha\ge \beta>0$ and $a,b>0$.  Let $t\to+\infty$. Let $\eta$, $\xi$ be independent  non-negative random variables. Assume that  
\begin{equation}\nonumber
\PP(\eta>t)=(a+o(1))t^{-\alpha}  
\qquad
{\rm{and}}
\qquad
\PP(\xi>t)=(b+o(1))t^{-\beta}.
\end{equation}
Put $c=a+b$ for $\alpha=\beta$, and 
$c=b$ for 
$\alpha>\beta$.
We have 
\begin{equation}\label{2017-03-30+}
\PP(\eta+\xi>t)=(c+o(1))t^{-\beta}.
\end{equation}
\end{lem}

For completeness, we present the proof Lemma \ref{LB2}.

\begin{proof}[Proof of Lemma \ref{LB2}] We first prove (\ref{2017-03-30+}) for $\alpha>\beta$.
Fix  $1>\gamma>\beta/\alpha$ and split the probability 
\begin{equation}\label{2017-03-30+++}
\PP(\eta+\xi>t)=\PP(\eta+\xi>t, \eta<t^{\gamma})+\PP(\eta+\xi>t, \eta\ge t^{\gamma})=:P_1+P_2.
\end{equation}
Here
\begin{eqnarray}\nonumber
P_1
&=&
\E
\bigl(\PP(\eta+\xi>t|\eta)
{\mathbb I}_{\{\eta<t^{\gamma}\}}
\bigr)
=
\E 
\bigl(
(b+o(1))(t-\eta)^{-\beta}
{\mathbb I}_{\{\eta<t^{\gamma}\}}
\bigr)
\\
\nonumber
&=&
(b+o(1))t^{-\beta}\PP(\eta<t^{\gamma})
=
(b+o(1))t^{-\beta}+O(t^{-\beta})\PP(\eta\ge t^{\gamma})
\\
\nonumber
&
=
&
(b+o(1))t^{-\beta}+o(t^{-\beta})
\end{eqnarray} 
and
\begin{equation}\nonumber
P_2\le \PP(\eta\ge t^{\gamma})=O(t^{-\gamma\alpha})
=
o(t^{-\beta}).
\end{equation}
Let us prove (\ref{2017-03-30+}) for $\alpha=\beta$. Fix $0.5<\gamma<1$.  Denote $\phi=\max\{\eta,\xi\}$,
$\tau=\min\{\eta,\xi\}$. We have
\begin{eqnarray}
\label{2017-03-30+1}
&&
\PP(\phi>t)=1-\bigl(1-\PP(\eta> t))(1-\PP(\xi>t)\bigr)
=
(a+b+o(1))t^{-\beta}+O(t^{-2\beta}),
\\
\nonumber
&&
\PP(\tau >t)=\PP(\eta> t)\PP(\xi>t)=O(t^{-2\beta}).
\end{eqnarray}
We write $\PP(\eta+\xi>t)=\PP(\tau+\phi>t)$ and proceed similarly as in (\ref{2017-03-30+++}):
\begin{equation}\label{2017-03-30++++}
\PP(\tau+\phi>t)=\PP(\tau+\phi>t, \tau<t^{\gamma})+\PP(\tau+\phi>t, \tau\ge t^{\gamma})=:P^*_1+P^*_2.
\end{equation}
Here
\begin{eqnarray}\nonumber
P^*_1
&
\le 
&
\PP(t^{\gamma}+\phi>t, \tau<t^{\gamma})
\le
\PP(t^{\gamma}+\phi>t)
=
\PP\bigl(\phi>t(1-o(1))\bigr)
=
(a+b+o(1))t^{-\beta}
\end{eqnarray} 
and
\begin{equation}\nonumber
P_2\le \PP(\tau\ge t^{\gamma})=O(t^{-2\beta\gamma})
=
o(t^{-\beta}).
\end{equation}
Finally, (\ref{2017-03-30+1}) and (\ref{2017-03-30++++}) 
imply 
\begin{displaymath}
(a+b+o(1))t^{-\beta}=\PP(\phi>t)
\le 
\PP(\tau+\phi>t)\le (a+b+o(1))t^{-\beta}.
\end{displaymath}
\end{proof}

%
%

\end{document}